\RequirePackage{fix-cm}
\documentclass[preprint,nonatbib]{elsarticle}
\journal{ }

\usepackage{geometry} 

\usepackage{paralist}

\usepackage{amsthm}
\usepackage{booktabs} %

\usepackage[ruled,
]%
{algorithm2e} %

\SetAlFnt{\small}
\SetAlCapFnt{\small}
\SetAlCapNameFnt{\small}
\SetAlCapHSkip{0pt}
\IncMargin{-\parindent}
\def\RCS$#1: #2 ${\expandafter\def\csname RCS#1\endcsname{#2}}
\RCS$Revision: 1.32 $
\RCS$Date: 2015/08/04 15:59:15 $ 

\usepackage{mathtools}
\usepackage{paralist}
\usepackage{xspace}
\usepackage{xcolor}
\usepackage{paralist}

\newcommand{\la}{\ensuremath{\leftarrow}}
\newcommand{\contract}{\mathbin{/}}

\usepackage{colortbl}

\usepackage{tikz}
\usetikzlibrary{arrows,matrix,positioning,decorations.pathreplacing}

\newcommand{\definedas}{\mathrel{\raise.095ex\hbox{:}\mkern-4.2mu=}}

\usepackage[numbers]{natbib}

\newtheorem{thm}{Theorem}%
\newtheorem{lem}[thm]{Lemma} 
 
\newtheorem*{claim}{Claim}
\newtheorem{corollary}{Corollary} 
\newtheorem{example}[thm]{Example}

\newtheorem{deff}{Definition}%
 
\theoremstyle{remark} 
\newtheorem*{rem}{Remark}

\newcommand{\st}{\textrm{s.t.\ }}

\newcommand{\dcup}{\dot\cup}

\usepackage[letterpaper=true]{hyperref}
\hypersetup{bookmarks=true,letterpaper=true,bookmarksnumbered=true,
   colorlinks=false,citecolor=blue,urlcolor=blue,linkcolor=blue,
   pdfauthor={Sven de Vries and Rakesh V. Vohra},
   pdftitle={Matroidal Approximations of Independence Systems},
   pdfsubject={RCS: \RCSRevision, \RCSDate},
   pdfkeywords={}}

\newcommand{\R}{{\mathbb R}}
\newcommand{\Rp}{{\mathbb R}_+}

\newcommand{\Zp}{{\mathbb Z}_+}

\newcommand{\bz}{\ensuremath{\mathbf 0}\xspace}
\newcommand{\bo}{\ensuremath{\mathbf 1}\xspace}

\newcommand{\CB}{\ensuremath{\mathcal B}}
\newcommand{\CM}{\ensuremath{\mathcal M}\xspace}
\newcommand{\CMp}{{\ensuremath{{\mathcal M}'}}\xspace}

\newcommand{\CI}{\ensuremath{\mathcal I}\xspace}

\newcommand{\CIp}{{\ensuremath{{\mathcal I}'}}\xspace}

\newcommand{\CJ}{\ensuremath{\mathcal J}}

\newcommand{\CC}{\ensuremath{\mathcal C}}

\newcommand{\CO}{\ensuremath{{\mathcal O}}}
\newcommand{\COp}{{\ensuremath{{\mathcal O}'}}}
\newcommand{\COd}{\ensuremath{{\mathcal O}^\dagger}}

\DeclareMathOperator{\supp}{supp}

\newcommand{\argmin}{\operatorname{arg\,min}}
\newcommand{\argmax}{\operatorname{arg\,max}}

\newcommand{\z}{^\top}

\usepackage{amsmath,amssymb}

\newenvironment{proofcl}{\begin{proof}}{\end{proof}}

\usepackage{amssymb,amstext,amsxtra}
\accentedsymbol{\dbarv}{\Bar{\Bar{v}}}

\usepackage{comment}
\excludecomment{WRONG}

\newcommand{\y}{\vphantom{1}\hphantom{x}} 

\begin{document}
\title{Matroidal Approximations of Independence Systems}
\author[1]{Sven de Vries\corref{cor1}}
\ead{devries@uni-trier.de}
\author[2]{Rakesh V. Vohra\fnref{fn1}}
\ead{rvohra@seas.upenn.edu}
\cortext[cor1]{Corresponding author}
\fntext[fn2]{Supported by NSF grant AST-1343381.}
\address[1]{Universit\"at Trier, Trier 54286, Germany}
\address[2]{University of Pennsylvania,
Department of Economics,
The Ronald O. Perelman Center for Political Science and Economics,
133 South 36th Street,
Philadelphia, PA 19104, USA}

\date{\today}

\begin{abstract}
    Milgrom (2017) has proposed a heuristic for determining a maximum 
    weight basis of an independence system $\CI$ given that we want an 
    approximation guarantee only for sets in a prescribed $\CO\subseteq \CI$. This $\CO$ reflects prior knowledge of the designer about the location of the optimal basis. 
The heuristic is based on finding an `inner matroid', one contained in 
    the independence system. We show that even in the case  $\CO=\CI$
    of zero additional knowledge the worst-case performance of this new 
    heuristic can be better than that of the classical greedy
    algorithm.
\end{abstract}
\maketitle

\section{Introduction}
Given a ground set $E=\{1,\ldots,n\}$ and
a family $\CI$ of subsets of $E$, the pair $(E, \CI)$ is a called an
independence system if $\emptyset \in \CI$ and for all $B \subseteq A
\in \CI$ we have $B \in \CI$ as well. Elements of $\CI$ are called
independent sets. For any $A \subseteq E$, a set $B \subseteq A $ is called
a basis of $A$ if $B \in \CI$ and $B \cup \{j\} \not \in \CI$ for all
$j \in A \setminus B$. 
For $A\subseteq E$ denote by $ \CB_\CI(A)$ (or $\CB(A)$ when no ambiguity) the set
of bases  of $A$ (with respect to $\CI$). The bases with respect
to $E$ are denoted by $\CB_\CI$ or just $\CB$.
Minimally dependent sets are called circuits. 
If we associate weights 
$v_i\in \Rp$ with each $i\in E$, the problem of finding a maximum
weight basis of $(E, \CI)$ can be expressed as $\max_{A \in
  \CB}\sum_{i \in A}v_i$. For convenience we will write $\sum_{i \in S}v_i$ as $v(S)$ for all $S
\subseteq E$. 
The problem is NP-hard by reduction to Hamiltonian path in a directed
graph \cite{korte-hausmann-1978}.

An independence system $(E,\CI)$ with $\{i\}\in\CI$ for all
$i\in E$ is called \emph{normal} (for graphs or matroids, the term
\emph{loopless} is more common). A non-normal independence system can be made normal by deleting the elements $e\in E$ with
$\{e\}\notin\CI$ without
changing the solution to the problem of finding a maximum weight
basis. %

\citet{milgrom-2017} proposed a heuristic for determining a maximum
weight basis given prior information about the optimal solution. This prior information is expressed as a subset $\CO$ of  $\CI$. The collection $\CO$ is to be interpreted to be the set of likely candidates for an 
optimal basis. The performance of Milgrom's heuristic is evaluated with respect to $\CO$. It is based on finding an `inner matroid', one contained in the independence system. We compare the worst-case performance of this heuristic with the well known greedy algorithm for the same problem. We show that the worst-case
performance of Milgrom's heuristic even in the absence of additional assumptions on
$\CO$  may be better than that of the classical greedy
algorithm. Additionally, as illustrated by example when one exploits $\CO$, there are aspects of Milgrom's proposal that bear further investigation. We defer a discussion of these matters till the end. In the next section we describe the greedy algorithm and state its worst case performance.

\section{Greedy Algorithm}
We describe the greedy algorithm (see algorithm~\ref{alg:GI} next page) for finding a basis of $E$ with possibly large weight

\begin{algorithm}
\caption{Greedy Algorithm \label{alg:GI} }
\SetKwInput{KwIn}{Given}
\SetKwInput{KwOut}{Goal}
\SetAlgoNoLine\LinesNumbered
  \KwIn{An independence system $(E, \CI)$ represented by an independence
    oracle.  Weights $v_i\in \Rp$ for all $i\in
    E$. %
  }

  \KwOut{An independent set $I_g \in \CI$\ of ``large'' value.}

  Order the elements of $E$ by non-increasing value $v_1 \geq v_2 \geq
  \ldots \geq v_n\geq 0$ \label{lGI-1}\;

  Set $I \la \emptyset$\;

  \For{$i\la 1$ \KwTo $n$}{

    \lIf{ $I\cup \{i\}$ independent (oracle call)}{ $I \la I\cup
      \{i\}$\label{lGI-5}} }

  Let $I_g \la I$ and return $I_g$.
\end{algorithm}

   \begin{rem}
     If some coefficients in the objective
     function are equal, the order of the $v_i$'s is not unique and
     therefore the algorithm's outcome is not necessarily unique. To
     ensure uniqueness assume an exogenously given tie breaking rule.   \end{rem}

 A worst case bound on the quality of the greedy solution in terms of the rank quotient of an independence system can be found in \cite{jenkyns-1976, korte-hausmann-1978}. 
\begin{deff}
Let $(E,\CI)$ be an independence system. The  \emph{rank} of   $F \subseteq E$
 is defined by
$r(F) \definedas \max \{|B| \, : \, B \text{ basis of } F\}$
and the \emph{lower rank} of $F$ is defined by
$l(F) \definedas \min \{|B| \colon B \text{ basis of }
F\};$
both map $2^E$ into the nonnegative integers.
The \emph{rank quotient} of
$\CI$ is denoted by 
$ q(\CI) \definedas \min_{F\subseteq E, r(F) > 0}\, \frac{l
(F)}{r(F)}.$  
\end{deff}
Following an axiomatization of matroids by
\cite{hausmann-korte-jenkyns-1980}, we define a matroid to be  an
independence system $(E,\CI)$ with  $q(\CI) = 1$.
There are several equivalent characterizations, some of
which we will use later:
\begin{description}
\item[Basis exchange:]  For every pair of bases $B^1,B^2\in\CB$ and
  $i\in B^2\setminus B^1$ there exists a $j\in B^1\setminus B^2$ such
  that $(B^1\cup \{i\})\setminus \{j\}$ is a basis.
\item[Augmentation property:] For every pair of independent sets
  $I,J\in\CI$ with $|I|<|J|$ there exists $j\in J\setminus I$ such
  that $I\cup\{j\}$ is an independent set.
\item[Rank axioms:] see~\eqref{R1}-\eqref{R3} on page~\pageref{R1}.
\end{description}
For proofs of these and more see \cite{oxley-b1992}.

If $(E, \CI)$ is a normal independence system, then,
$q(\CI)\in\left\{\frac{a}{b}\mid a\in\{1,2,\dots,l(\CI)\} \text{ and } b\in\{1,2,\dots,r(\CI)\}\right\}.$
Hence, for normal independence systems, $q(\CI)\geq 1/r(\CI).$
However, a better bound is known:
\begin{thm}[\citealp{hausmann-korte-jenkyns-1980}]\label{thm:hkj}
  Let $(E,\CI)$ be an independence system. If, for any $A\in\CI$ and
  $e\in E$ the set $A\cup\{e\}$ contains at most $p$ circuits, then
  $q(\CI)\geq 1/p$. 
If $(E,\CI)$ is the intersection of $p$ matroids,
  then,  $q(\CI)\geq 1/p$. 
\end{thm}

\begin{thm}[\citealp{jenkyns-1976,korte-hausmann-1978}] \label{satz:5.18}
  For an independence system $\CI$ on
$E$ with objective function $v\gneqq \bz$, let $I_g$ be
the solution returned by Algorithm~\eqref{alg:GI} and $I_o$ be a maximum weight
basis. Then,
$q(\CI) \leq \frac{v(I_g)}{v(I_o)} \leq 1.$ 
Also, for every independence system there are weights $v \in
\{0,1\}^E$, such that the first inequality holds with equality.
\end{thm}

The only attempt to improve upon this bound we are aware of involves
incorporating a partial enumeration stage into the greedy algorithm, see~\cite{hausmann-korte-jenkyns-1980}.

\section{Inner Matroid}
\citet{milgrom-2017} proposes an alternative to the greedy
algorithm with two parts. First, is the introduction of
\emph{a-priori} information on where the optimal basis may lie. Formally,  let $\CO\subseteq
  \CI$ be a collection of independent sets one of which is conjectured
  to be an optimal weight basis. $\CO$ need not satisfy the hereditary
  property (so $J\subset I\in \CO$ does not imply $J\in \CO$) and may exclude a basis that is in fact
  optimal (see Example~\ref{ex:stab}). Call $\CO$ the \emph{acceptable set}. The second part
defines a matroid `inside' the independence system (but not necessarily
containing all of $\CO$) and applies the greedy algorithm to find an optimal weight basis of that matroid. We describe this approach here.

Given an independence system $(E, \CI),$ an acceptable set $\CO$,  and a weight vector $v\in\Rp^E$ let
  \[V^*(\CI,v)=\max_{S\in\CI}v(S),\quad V^*(\CO, v) =\max_{S\in\CO}v(S),
\text{ and }V^*_g(\CI,v)=\text{value of the greedy solution}.\]
Independence system $(E,\CM)$ is a matroid by the previous definition
involving rank quotient and Theorem~\eqref{satz:5.18} if and only if 
  $V^*(\CM,v)=V_g^*(\CM,v)$ for all $v\in\Rp^E.$

Call an independence system $(E,\CJ)$  contained in
$(E,\CI)$, an inner independence system, if $\CJ\subseteq \CI$.
If $(E, \CJ)$ is a matroid it is called an inner matroid.

 The \emph{approximation quality of $\CJ$ for $\CI$ with respect to the acceptable set $\CO \subseteq \CI$} is given by
\[\rho(\CI,\CO,\CJ)\definedas 
\min_{S \in \CO\setminus\{\emptyset\}}\min_{S'\in \CB_\CJ(S)}
  \frac{|S'|}{|S|};\]
so this picks the element from $\CO$ that is least well approximated by a
basis of $\CJ$. In terms of the lower rank $l_\CJ$  of $(E,\CJ)$ this yields
\[\rho(\CI,\CO,\CJ)=
\min_{S \in \CO\setminus\{\emptyset\}}
  \frac{l_\CJ(S)}{|S|}.
\]
 If $\CM$ is an inner matroid, then $r_\CM=l_\CM$ and
\[\rho(\CI,\CO,\CM) = \min_{S \in
  \CO\setminus\{\emptyset\}}
  \frac{r_{\CM}(S)}{|S|}
=\min_{S \in
  \CO\setminus\{\emptyset\}}\max_{S'\in\CM:S'\subseteq
  S}\frac{|S'|}{|S|},\]
which is the form \cite{milgrom-2017} originally chose. 

It is natural to ask when $\rho(\CI,\CO,\CM)=1$.

\begin{thm}\label{rho1}
If $(E, \CI)$ is a normal independence system  with acceptable set $\CO\subseteq\CI$, then, there exists a matroid $\CM$ with
  $\rho(\CI,\CO,\CM)=1$ if and only if \CI\ has an inner matroid
  containing $\COd\definedas
  \bigcup_{F\subseteq B\in\CO}\{F\}$.  
\end{thm}

\begin{proof}
 If $\rho(\CI,\CO,\CM)=1$ then $\CO\subseteq \CM$. For $I\subseteq J\in
 \CO$ follows $J\in \CM$ and by hereditary property of matroid $I\in\CM.$
\end{proof}

We
give three examples to suggest that nothing stronger is possible. 
\begin{deff}
Let $U^k_n$ denote the \emph{uniform matroid} of all subsets of an $n$-element
    set of cardinality at most $k$, where  $0\leq k\leq n$. If
    we want to specify the ground set $E$ of $n$ elements explicitly
    we can also write $U^k_E.$
\end{deff}

\begin{example}
In all cases the ground set is $E = \{1,2,3,4\}$.
  \begin{itemize}
  \item Let $\CI = \{\emptyset,\{1\},\{2\},\{3\},\{4\}, \{1,2\},\{3,4\}\}$. $(E,
    \CI)$ is not a matroid as the basis exchange axiom is violated by
    $\{1,2\}$ and $\{3,4\}$. If
    $\CO=\{\{1\},\{2\},\{3\},\{4\}\}$, then, $\rho(\CI,\CO,U^1_4)=1$.
    This demonstrates, that the condition ``$\CO$ is a matroid'' is
    too strong. However, in this case \COd\ \emph{is} a matroid.
  \item %
    Let $\CI$ be as above, but    $\CO=\{\{1,2\},\{3,4\}\}$. Then, $\rho(\CI,\CO,U^1_4)=1/2$.
    In this case \COd\ is not a matroid.
\nopagebreak
  \item Let $\CI = \{\emptyset\{1\},\{2\},\{3\},\{4\},\{1,2\}, \{2,3\}, \{3,4\}, \{1,3\}, \{1,4\}, \{2, 4\}, \{1,2,3\}\}.$
    $(E, \CI)$ is not a matroid, since $\{3,4\}$ is a independent set
    that can not be augmented with an element from the larger
    independent set $\{1,2,3\}$. If $\CO=\{\{2,3\},\{3,4\}\}$ the smallest inner
    matroid of $\CI$ containing $\CO$ is
    $\CM=U^2_{\{2,3,4\}}$;  and 
    $\rho(\CI,\CO,\CM)=1.$ However $\COd$ is not a matroid, since it
    lacks the basis $\{2,4\}$ that would be required by basis exchange.
    So requiring ``$\COd$ to be a matroid'' is too strong.
  \end{itemize}
\end{example}

\begin{rem} Clearly
  $\rho(\CI,\CO,\CJ)\geq\rho(\CI,\COd,\CJ),$
  since the minimum
  on the right hand side is determined over a larger set.
\end{rem}

However, $\rho(\CI,\CO,\CJ)>\rho(\CI,\COd,\CJ)$ is possible, as the
following example demonstrates:
\begin{example}
  Consider $E=\{1,2,3,4,5,6\}$ and the independence system
  $\CI=2^{\{1,2,3,4\}}\cup 2^{\{3,4,5,6\}}.$ Clearly $\{1\},
\{5,6\} \in \CI$. The augmentation property 
  requires that  $\{1,5\}$ or $\{1,6\}$  has to be independent, which
  is not the case. Therefore \CI\ is not a matroid.
  Let $\CO = \{\{1,2,3,4\}, \{3,4,5,6\}\}$ be our acceptable set and let
  $\CM=2^{\{1,2,3,4\}}$.
  Now, let us determine $\rho(\CI,\CO,\CM) = \min_{S \in
  \CO\setminus\{\emptyset\}}
  \frac{r_{\CM}(S)}{|S|}.$  With $r_\CM(\{1,2,3,4\})=4$ and $r_\CM(\{3,4,5,6\})=2$ %
  we deduce that $\rho(\CI,\CO,\CM)=\frac24=\frac12.$
  
  Now, $\COd$ contains $O=\{4,5,6\}$ with
  $r_\CM(O)=1$. Even worse $O=\{6\}\in\COd$ with $r_\CM(O)=0$ which
  demonstrates $\rho(\CI,\COd,\CM)= 0<\frac12=\rho(\CI,\CO,\CM)$.
\end{example}
\begin{corollary}
  The example demonstrates, that %
  $\rho(\CI,\COd,\CM)>0$ is only possible if 
  $e\in\bigcup_{O\in\CO}O$ implies $\{e\}\in \CM.$
\end{corollary}

\begin{example}[twin-peaks]\label{ex:twin-peaks}
  Consider two disjoint sets $E_1$ and $E_2$ and integers $k_1,k_2$
  such that $k_1<|E_1|<k_2<|E_2|.$ Let $\CM_1=U^{k_1}_{E_1}$ and $\CM_2=U^{k_2}_{E_2}$.

  Define $(E, \CI)$ to be an
  independence system on $E\definedas E_1\dcup E_2$ where a set $A
  \subseteq E$ is independent if 
  and only if $A$ is an independent set in $\CM_1$ \emph{or} $\CM_2$, thus
  $\CI=\CM_1\cup\CM_2$. Call this a `twin-peaks' independence system. %
  Choose an $F\subseteq E$ which consists of one element from $E_1$
  and $k_2$ elements from $E_2$. Now $l(F)=1$ and $r(F)=k_2$. Since larger
  ratios between sets in this independence system are impossible, $q(\CI)=\frac1{k_2}$.

  To determine the best matroid approximating $\CO=\CI$ contained in \CI, we must examine all matroids $\CM'$
  contained in $\CI:$
  \begin{enumerate}
  \item Suppose first that there exists $i\in E_1, j\in E_2$ such that $\{i\}, \{j\}\in \CMp$. We show that the rank of $\CMp$ is 1, i.e. the largest bases of $\CMp$ are of size one. If not, there exists a set $F \in \CMp$ of cardinality two. By construction of $\CI$
    the set $F$ must be a subset of $E_1$ or $E_2$, wlog. let
    $F\subseteq E_1$; since $F$ and $\{j\}$ are independent sets in
    $\CMp$ there has to be an element in $F$ with which we could
    augment $\{j\}$ to be an independent set $F'$ in $\CMp$. But $F'$
    contains elements from $E_1$ and $E_2$ and therefore is not
    independent in $\CI$ and therefore not in $\CMp.$ 

    The largest matroid of this kind contained in \CI\ is $U^1_E$.
Hence, every independent set from \CI\ is approximated by an
    arbitrary contained singleton and we obtain the worst case bound of
    $1/\max(k_1,k_2)=1/k_2.$

  \item On the other hand, if $r(\CMp)>1$ then we can conclude
    that either all independent sets of \CMp\ are contained in $E_1$ or
    they are contained in $E_2$  and the inclusion-wise largest
    matroids contained in $E_1$ and $E_2$ are $\CM_1$ and $\CM_2,$ respectively.

    For $\CMp=\CM_1$, approximating sets $F\in\CM_2\setminus\{\emptyset\}$
    is only possible with the 
    empty set which yields a quotient of $0/|F|$. For
    $\CMp=\CM_2$, approximating sets
    approximating nonempty sets $F\in\CM_1\setminus\{\emptyset\}$
    is only possible with the 
    empty set which yields a quotient of $0/|F|.$
  \end{enumerate}
In summary we obtain the best approximation in the first case and can
conclude
\[\rho(\CI,\CI,\CMp)=\frac{1}{k_2}=q(\CI).\]

\textbf{However,} $\rho$ can be much better if we exploit additional
knowledge.  As an illustration, suppose that $\CO= \binom{E_2}{k_2}$,
i.e., all subsets of size at most $k_2$ of $E_2$.
This might be justified by a situation, where we know
that the optimal solutions which occur in our environment, have to
live in $E_2$, maybe since $k_1\ll k_2$.
Unsurprisingly,  the matroid $\CM_2$ gives us the best conceivable
guarantee of $1$:
\[\rho\left(\CM,\binom{E_2}{k_2},\CM_2\right
)=1>\frac{1}{k_2}=q(\CI).\]
\end{example}

\citet{milgrom-2017} proposes that the maximum weight basis of an
inner matroid $\CM$ be used as a solution to the problem of finding a
maximum weight basis in $(E, \CI)$. %
The objective function value of this solution will be
$V^*(\CJ,v)$. \cite{milgrom-2017}  gives a bound on the quality of
this solution in terms of $\rho(\CI,\CO,\CJ)$. 

\begin{deff}
  For a given independence system $(E,\CI)$ and $e\in E$ define a
  new independence system $\CI\setminus\{e\}$  on
  $E\setminus \{e\}$ by
  $\CI\setminus\{e\}\definedas \{I\in\CI\colon e\notin I\}.$
  Additionally, for $\{e\}\in\CI$ define the independence system
  $\CI\contract\{e\}$ on $E\setminus \{e\}$ by
  $\CI/\{e\}\definedas\{I\setminus\{e\}\colon e\in I\in\CI\}.$
\end{deff}

\begin{rem}
  It is easy to see, that $(E\setminus\{e\},\CI\setminus\{e\})$ is an
  independence system and $\CI\contract\{e\}$ is an independence
  system if $\{e\}\in\CI$. Furthermore, if $(E,\CI)$ is a matroid,
  so is $(E\setminus\{e\},\CI\setminus\{e\})$.
\end{rem}

\begin{lem}
  Given an independence system $(E,\CI)$ and a set $\CO$ the set 
  \[W\definedas \{ w\in \R^E_+\colon \bo\z w=1\} \cap\{w\in\R^E_+\colon \emptyset\neq\CO\cap
  \argmax_{I\in \CI} w(I)\}\]
  is compact.
\end{lem}
\begin{proof}
  Clearly the set $\{ w\in \R^E_+\colon \bo\z w=1\}$ is compact. Also,
  \begin{align*}
    \{w\in\R^E_+\colon \emptyset\neq\CO\cap  \argmax_{I\in \CI} w(I)\}
    &=\bigcup_{O\in\CO} \{w\in\Rp^E\colon O\in \argmax_{I\in \CI} w(I)\},
  \end{align*}
  where each set on the right is a closed cone. As $\CO$ is finite the
  right hand side is closed. Therefore $W$ is compact.
\end{proof}

\begin{thm}[\citealp{milgrom-2017}]\label{miltheorem} 
Let $(E,\CM)$ be an inner matroid and $W$ be the above defined set
of non-negative and non-trivial weight vectors $v$ such that at
least one member of $\CO$ is a maximum weight basis of $(E, \CM)$. %
Then,
  \[\rho(\CI,\CO,\CM) =\min_{v\in W} \frac{V^*(\CM,v)}{V^*(\CO,v)}.\]

\end{thm}

As a convenience for the reader we provide here an independent, fuller
proof of the result than \cite{milgrom-2017}.
\begin{proof}%
If the theorem is false, there must exist an independence system $(E,
\CI)$, a set $\CO$ and an inner matroid $(E, \CM)$ such that
\begin{equation}
\rho(\CI,\CO,\CM) > \min_{v\in W} \frac{V^*(\CM,v)}{V^*(\CO,v)}.\label{eq:viol}
\end{equation}

Among all such counterexamples choose one where $|E|$ is minimal.
Among all such counterexamples choose one, where the support of
$v$ is smallest. Denote the size of the support by $\ell$.

Let
\begin{equation}
 v^*\in \argmin_{v \in W:\supp(v)\leq\ell}
  \frac{V^*(\CM,v)}{V^*(\CO,v)},\label{eq:v*}
\end{equation} 
such that it is a vector with as few distinct values as possible. 
Choose $X_\CO\in\argmax_{S\in\CO}v^*(S),$ and
$ X_{\CM}\in\argmax_{S\in\CM}v^*(S)$ such that $|X_\CO\cap X_\CM|$ is maximal.
In case of ties, choose a tie-breaking rule, so that greedy yields $X_\CM$.
By the minimum support assumption, $v^*_j = 0$ for all $j \not \in X_{\CM} \cup X_{\CO}$.
Additionally:

\begin{claim}
  $v^*_j>0$ for all $j\in E\setminus X_\CO.$
\end{claim}
\begin{proofcl}
  Suppose there exists $j\in E\setminus X_\CO$ with $v^*_j=0.$
  Let $\CIp=\CI\setminus\{j\}$, $E'=E\setminus\{j\}$,
  $\CMp=\CM\setminus\{j\}$ and
  $\COp= \{O\in\CO\colon j\notin O \}.$
  Because $j\notin X_\CO\in \COp$ we have $\COp\neq\emptyset$.

  For our counterexample we have by~\eqref{eq:viol} and~\eqref{eq:v*}:
  \[\min_{S\in\CO\setminus\{\emptyset\}}\frac{r_\CM(S)}{|S|}%
  >
  \frac{V^*(\CM,v^*)}{V^*(\CO,v^*)}.\]
  Now, since the minimum on the left below is taken over a smaller set and
  $r_\CMp(S)=r_\CM(S)$ for all $S\in\COp$,
  \[\min_{S\in\COp\setminus\{\emptyset\}}\frac{r_\CMp(S)}{|S|}\geq
  \min_{S\in\CO\setminus\{\emptyset\}}\frac{r_\CM(S)}{|S|},\]

  On the other hand, $V^*(\CM,v^*)\geq V^*(\CMp,v^*)$ and
$V^*(\CO,v^*)=V^*(\COp,v^*)>0$ yield
    \[
    \frac{V^*(\CM,v^*)}{V^*(\CO,v^*)}\geq \frac{V^*(\CM',v^*)}{V^*(\CO',v^*)}.\]
    Together
  \[\min_{S\in\COp\setminus\{\emptyset\}}\frac{r_\CMp(S)}{|S|}=\rho(\CIp,\COp,\CMp)>
  \frac{V^*(\CMp,v^*)}{V^*(\COp,v^*)}.\]
  Therefore, we obtain a counterexample with $|E'|<|E|$ in
  contradiction to our assumption of minimality of the counterexample.
\end{proofcl}

\begin{claim}
 $v^*_k=0\,\, \forall k\in X_\CM\setminus X_\CO$.
\end{claim}
\begin{proofcl}
  Suppose a $k\in X_\CM\setminus X_\CO$ with $v^*_k>0$ exists. Let
  $e^k$ denote the $k$-th unit vector.
  Consider the valuations $v^\alpha=v^*-\alpha e^k$ for
  $0\leq\alpha< v^*_k$ and their
  normalization
  \[\hat v^\alpha = \frac{1}{\sum_{i\in E} v^\alpha_i} v^\alpha \quad\text{
    and observe }\quad \frac{\hat v^\alpha (S)}{\hat v^\alpha(T)}=\frac{v^\alpha
  (S)}{v^\alpha(T)}\; \forall S,T:v(T)>0.\]

Now  $\hat v^\alpha$ differs from $v^*$ in two ways. First,  the
weight of one element outside $X_\CO$  was reduced by
$\alpha$. Second, the weights were scaled.  Hence, $X_\CO$ is optimal
with respect to $v^\alpha$ for all $\alpha\geq 0$.

If there exists no $\epsilon>0$ such that $X_\CM$ remains optimal for $\CM$ with respect
to $v^\alpha$ for $0\leq \alpha\leq\epsilon$, then it is the case that
$X_\CM$ is not optimal for any $v^\epsilon$ with $\epsilon>0.$ This requires, that the
run of the greedy algorithm changes if $\epsilon>0$. However, the
greedy algorithm does not depend on the exact weights of $v^\epsilon$
but only on their ordering with respect to decreasing weight.
The only change in ordering is, that with respect to $v^\epsilon$ for
$\epsilon>0$ that item $k$ will occur later. So before the greedy
algorithm for $v^\epsilon$ considers $k$ it picks up some item $k'$
that is not part of $X_\CM.$ Now, as $k'$ was not choosen for $v^0$,
item $k'$ must be dependent with respect to the previously choosen items and
$k$.  Hence there exists a circuit in $X_\CM\cup k'$ that contains $k$
and $k'.$ Now, as this happens for arbitrarily small $\epsilon$ it has
to be the case, that $v_k=v_{k'}$ (and all elements choosen inbetween
have that same weight). But this yields, that $X_\CM\cup k'$ is
dependent while $X_\CM\setminus \{k\}\cup\{k'\}$ is another basis and
$k'\in X_\CO\setminus X_\CM$ which contradicts our
assumption, that $X_\CM$ and $X_\CO$ were chosen to have a maximal
intersection.

  Therefore, such an $\epsilon>0$ has to exist so that $X_\CM$ stays
  optimal for $v^\alpha$ with $0<\alpha<\epsilon$.
  Now, observe $v^\alpha(X_\CM)$ and $\hat v^\alpha(X_\CM)/\hat v^\alpha(X_\CO)$
  are decreasing for $\alpha$ increasing 
  from $0$ to $\epsilon$. As $|E|$ and $|\supp v|$ do not change this contradicts our choice of $v^*$ in \eqref{eq:v*}.
\end{proofcl}
Together with the previous claim we obtain:
\begin{claim}
  $X_\CM\setminus X_\CO=\emptyset$, therefore $X_\CM\subseteq X_\CO$
  and $X_\CO=E$.
\end{claim}

\begin{claim}
  For every $i\in X_\CM$ there exists an $j\in X_\CO\setminus X_\CM$
  such that $v^*_i=v^*_j$ and vice-versa.
\end{claim}
\begin{proofcl}
  Consider a valuation where $v^*_i$ is slightly decreased. If $X_\CO$
  and $X_\CM$ do not change, then a smaller quotient~\eqref{eq:viol} for that
  valuation would result, which is a contradiction. As $X_\CO=E$ that cannot change either; so $X_\CM$ has to change for arbitrarily small
  changes to $v^*_i$, and it can change only, by dropping $i$ and
  picking up some new element $j\in X_\CO\setminus X_\CM.$ But this
  requires $v^*_i=v^*_j$.

  Similarly, consider a valuation, where $v^*_j$ is slightly increased. If $X_\CO$
  and $X_\CM$ do not change, then a smaller quotient for that
  valuation would result, which contradicts assumption. Again,
  $X_\CO=E$ ensures that $X_\CO$ remains optimal during these slight increases. However, it
  could be, that $j$ immediately enters $X_\CM$ but this requires, that
  just as immediately some $i'\in X_\CM$ exits $X_\CM$. But this
  requires $v^*_{i'}=v^*_j$.
\end{proofcl}
Let $\bar v=\max_{i\in E} v^*_i$, and
 $\dbarv\definedas \max_{i:v^*_i\neq \bar v}
v^*_i, $ %
 and set
\[\hat v_j^{\alpha}\definedas
\begin{cases}
  \bar v + \alpha&\text{ if } v^*_j=\bar v,\\
  v^*_j&\text{ otherwise.}
\end{cases}\]
If $\dbarv=0$ then  $v^*$ is a multiple of the all ones vector, and the
proof would be done, as such a $v^*$ could not provide a counterexample.

Otherwise, if $\dbarv>0$  for $|\alpha|\leq \bar v-\dbarv$,  $X_\CO$ and $X_\CM$ remain optimal for $V^*(\CO,\hat v_j^\alpha)$
and $V^*(\CM,\hat v_j^\alpha)$ respectively, since $X_\CO=E$ and
$v\geq 0$ implies, that $X_\CO$ is best possible.
Regarding the optimality of $X_\CM$, notice, that the greedy algorithm 
finds an optimal solution, but for $|\alpha|\leq \bar v-\dbarv$ the
ordering of the elements does not change, hence $X_\CM$ remains
optimal.
Let 
\[f(\alpha) \definedas\frac{V^*(\CM,\hat v_j^\alpha)}{V^*(\CO,\hat
  v_j^\alpha)}.\]
By our analysis we deduce 
\[f(\alpha) \definedas\frac{V^*(\CM,\hat v_j^*)+\alpha|\{i\in
  X_\CM\colon \bar v=v^*_i\}|}{V^*(\CO,\hat
  v_j^\alpha)+\alpha|\{i\in
  X_\CO\colon \bar v=v^*_i\}|}
.\] 
 If $f'(0) < 0$, then $f(\epsilon)<f(0)=
\frac{V^*(\CM,v^*)}{V^*(\CO,v^*)}$ (for small $\epsilon$) contrary to
the assumption, that 
$v^*$ was chosen best possible among those with minimal $|E|$ and
support of size $\ell$. For small $\epsilon$ neither $E$ nor
$\supp(v^\epsilon)$ have changed.

If $f'(0) > 0$, then $f(-\epsilon)<f(0)=
\frac{V^*(\CM,v^*)}{V^*(\CO,v^*)}$ (for small $\epsilon$) contrary to
the assumption, that 
$v^*$ was chosen best possible among those with minimal $|E|$ and
support of size $\ell$. For small $\epsilon$ neither $E$ nor
$\supp(v^{-\epsilon})$ have changed.

Finally, if $f'(0) = 0$,  it is straightforward to show that $f$ is a constant function for
$\alpha\in[\dbarv -\bar v,\bar v -\dbarv]$,
hence $f(\dbarv -\bar v)=f(0)$ but
the number of distinct values of $\hat v^{\bar{\bar v} -\bar v}$ has
decreased by $1$ contrary to the assumption that $v^*$ was chosen
to have the minimum  number of distinct values.

Since in all three cases we got a contradiction, we conclude that
$v^*$ is just a multiple of the all-ones vector.
\end{proof}

\begin{rem}
  Uniqueness of  $X_\CO$ and $X_\CM$ is important for the proof.
  Consider a twin-peaks independence system consisting 
  of $\{1\}$ and of $\{2,3\}$. The bases are $\{1\}$ and $\{2,3\}$.
  Let $2,1,1$ be the coefficients of the objective function.
  Optimal are both $X_1=\{1\}$ and $X_2=\{2,3\}.$

  If we choose $X_2$ for the optimal solution and $\alpha$ increases
  $v_1$ by $0$ upwards then the objective function is $\alpha |X_1|$ which
  does not relate to $X_2$.

    If we choose $X_1$ for the optimal solution and $\alpha$ decreases
  $v-1$ by $0$ downwards, then, the objective function is constantly $|X_2|$ which
  does not relate to $X_1$.
\end{rem}

Theorem~\ref{miltheorem} about approximating an independence system by
a matroid motivates the search for a matroid that yields
a best approximation. To this end, \citet{milgrom-2017} defines the following \emph{substitutability index}:
\begin{equation}\label{rho}
\rho^M(\CI,\CO)\definedas\max_{\CM \subseteq \CI: \CM \text{ is a
      matroid}}\rho(\CI,\CO,\CM).
  \end{equation}

\citet{milgrom-2017} gives no algorithm for determining the best inner
matroid $\CM'$. We describe an integer program %
for finding it.

Let $r$ be the rank function of the independence system $(E, \CI)$. We use the fact that matroids are completely characterized by their rank functions. A function $r:
  2^E\to \Zp$ is the rank function of matroid  $(E,\CI)$,
    with
    $\CI=\{F\subseteq E\colon r(F)=|F|\}$ if and only if for all
 $X,Y\subseteq E$:
\begin{align}
    \tag{R1}\label{R1} & r(X)\leq |X|\\
    \tag{R2}\label{R2} & \text{If } X\subseteq Y, \text{ then }
    r(X)\leq r(Y)\\
    \tag{R3}\label{R3} & r(X\cup Y)+r(X\cap Y)\leq r(X)+r(Y)%
  \end{align}

We formulate the problem of finding the inner matroid as the problem of finding a suitable rank function, $r'(S) = r(S) - \delta_S$, with $\delta_S$, integral,  must be chosen to ensure that $r'$ is a matroid rank function. Hence,
\begin{align}
\begin{alignedat}{2}
  \rho^M(\CI, \CO) = \max_{\delta} \min_{S \in \CO} (1- \frac{\delta_{S}}{r(S)})&\\
\st\quad 0 \leq \delta_{S} &\leq r(S)&& \forall S \subseteq E\\
[r(S)-\delta_{S}]  &\leq [r(T)-\delta_{T}]&&\forall S
\subset T\subseteq E\\
[r(S \cup j \cup k) - \delta_{S \cup j \cup k}] - [r(S \cup j ) - \delta_{S \cup j}]&\leq [r(S \cup k ) - \delta_{S \cup k}] - [r(S) - \delta_{S}]\quad&&\forall j, k \not \in S \subseteq E\\
\delta_S &\text{ integral }&& \forall S \subseteq E\\
\end{alignedat}\tag{IP}\label{IP}
\end{align}

Obviously, this is not  an efficient means to find  an optimal  inner
matroid, since it is an integer-program %
Additionally, the usually given independence
oracle is of no use either.
However, if one  faces  %
a single independence system with possibly different objective
functions but always, for reasons of fundamental value, the same
acceptable set $\CO$, then finding a good inner matroid %
is just preprocessing.

\section{Possible advantages over  Greedy}
In this section we compare the outcome of the algorithm in
\citet{milgrom-2017} with greedy in the zero knowledge case of $\CO =
\CI$. Example~\ref{ex:twin-peaks} investigated a case, in which under
zero knowledge the usual rank-quotient bound could not be
improved. This does not exclude the possibility that there exist  independence systems $(E, \CI)$ where $\rho^M(\CI, \CI) > q(\CI)$. 

\begin{thm}\label{thm:smart_zero_knowledge}
  There exists an independence system $(E,\CI)$ such that $\rho^M(\CI, \CI) > q(\CI)$.
\end{thm}
\begin{proof}
  Let $E=\{-2,-1,1,2,3,4\}$ and 
  $\CI=\left\{ I\subseteq E\colon \text{ either } I=\{-2,-1\} \text{ or } |I\cap \{-2,-1\}|\leq 1\right\}.$

  On the one hand the set $\{-1,1,2,3,4\}$ is independent and maximally
  independent, hence a basis of $E$. On the other hand the set
  $\{-2,-1\}$ is independent and cannot be augmented; hence it is a
  basis of size $2$. Notice, that all two element sets are
  independent in \CI. Therefore we can 
  conclude $l(E)=2$ and $r(E)=5$ and $q(\CI)\leq 2/5$. 

  By the earlier remark we made  for normal independence systems
  $q(\CI)\in\left\{\frac{a}{b}\mid a\in\{1,2,\dots,l(\CI)\}\right.$
  and \linebreak $\left.b\in\{1,2,\linebreak\dots,r(\CI)\}\right\}$.
  So if the rank quotient were smaller than $2/5$ for some set $F$ we
  would need $r(F)>2$. Hence $|F|\geq 3$. Recall that all pairs of elements are independent, hence the
  lower rank of any set of size at least $2$ is at least
  $2$. Therefore, $q(\CI)=2/5.$
  
  Now consider the almost free matroid \CM on set $E$ given by 
  $\CM=\left\{ I\subseteq E\colon |I\cap \{-2,-1\}|\leq 1\right\}.$
  As it it contains one `OR'-clause less than $\CI$, clearly,
  $\CM\subseteq \CI.$ 
  Additionally we note that $\CM=U^1_{\{-2,-1\}}\oplus
  U^4_{\{1,2,3,4\}}$, hence \CM is a matroid, in fact an inner
  matroid of \CI.

  Further, by construction $\CM=\CI \setminus \{\{-2,-1\}\}$, i.e., all sets of
  $\CI$ but $\{-2,-1\}$ are independent in $\CM$ and therefore
  perfectly approximated. Now, clearly the best approximations of the
  set $\{-2,-1\}$ from \CI are $\{1\}$ or $\{2\}$ both of
  approximation quality $1/2$. Therefore $\rho(\CI,\CI,\CM)=1/2$ and
  $\rho^M(\CI,\CI)\geq \rho(\CI,\CI,\CM)=1/2.$\footnote{To argue that
    $\rho^M(\CI,\CI)=1/2$ we prove, that there exists no matroid
    \CMp with $\rho(\CI,\CI,\CMp)>1/2$. For a contradiction suppose,
    inner matroid $\CMp\subset\CI$ fulfills $\rho(\CI,\CI,\CMp)>1/2$. As it has to
    approximate the set $L=\{-2,-1\}$ better than $1/2$, we can
    conclude that $L\in\CMp$. By ratio assumption and
    $I=\{-1,1,2,3,4\}\in \CI$  we know that \CMp contains a set of $3$
    elements of $I$. Wlog. we can assume that that set is either of
    form $J_1=\{-1,1,2\}$ or $J_2=\{2,3,4\}.$ Now applying the
    augmentation property of matroids to $L$ of size $2$ and $J_1$ or
    $J_2$ of size $3$ shows that $\CMp$ contains an independent set
    $M$ of
    size $3$ containing $-2,-1$ and one element of $1,2,3,4$. But by
    definition, such $B$ is not contained in $\CI$ and therefore note
    in $\CMp.$} This demonstrates 
  \[\rho^M(\CI, \CI) \geq 1/2> 2/5=q(\CI).\]
Hence, a well chosen inner matroid (following Milgrom) can provide
  a better performance guarantee than direct application of the greedy algorithm
  to an independence system.
\end{proof}

This independence system may appear contrived.
This is not the case. It is the independence system of a
knapsack problem with capacity $8$ and  a weight of $4$ for items $-2,-1$
and weight of $1$ for items $1,2,3,4$. 

There is a classic approximation
algorithm for the knapsack problem. It generates two solutions and picks the best of them. One of those solution is a greedy solution (according to value per weight), and  has a performance guarantee of
$1/2$. Theorem \ref{thm:smart_zero_knowledge} is no stronger than the
classic approximation algorithm for knapsack. Therefore, we provide a knapsack instance
and an inner matroid of its independence system with an approximation
quality of better than $1/2$.

\begin{thm}\label{thm:knapsack}
  For the independence system $(\CI,E)$ induced by the knapsack
  problem on items $E=\{-3,-2,-1, \linebreak 1,2,3,4,5\}$
  with weights $w_i=5$ if $i<0$ and $w_i=1$ if $i>0$ and capacity
  $15$ holds  $\rho^M(\CI, \CI) = 2/3>1/2>3/7=
  q(\CI)$. Therefore, there exists a matroid that permits
  $2/3$-approximation and is hence better than the
  knapsack-greedy algorithm. 
\end{thm}

\begin{proof}
  Consider the knapsack problem on items $E=\{-3,-2,-1, 1,2,3,4,5\}$
  with weights $w_i=5$ if $i<0$ and $w_i=1$ if $i>0$ and capacity
  $15$.
  It induces the independence system 
  \[\CI=\left\{ I\subseteq E\colon \text{ either } I=\{-3,-2,-1\} \text{ or } |I\cap \{-3,-2,-1\}|\leq 2\right\}.\]

  On the one hand the set $\{-2,-1,1,2,3,4,5\}$ is independent and maximally
  independent, hence a basis of $E$. On the other hand the set
  $\{-3,-2,-1\}$ is independent and cannot be augmented; hence it is a
  basis of size $3$. Notice, that any triple of elements is
  independent in \CI. Therefore we can 
  conclude $l(E)=3$ and $r(E)=7$ and $q(\CI)\leq 3/7$. 

  By earlier remark we had for normal independence systems  $q(\CI)\in\left\{\frac{a}{b}\mid
    a\in\{1,2,\dots,l(\CI)\}\right.$ and
    $\left.b\in\{1,2,\dots,r(\CI)\}\right\}$; together with the observation
  that all triplets are independent in this $\CI$  it
  reduces to $q(\CI)\in\left\{\frac{a}{b}\mid
    a\in\{1,2,\dots,3\}\right.$ and
    $\left.b\in\{1,2,\dots,7\}\right\}.$
  If there were $F\subseteq E$ with $q(F)=l(F)/r(F)=a/b<3/7$ then because of
  $b\leq 7$ we need $a\leq 2$. But as triplets are independent,
  $b=l(F)\leq 2$ requires $|F|\leq 2$ which yields $a=r(F)\leq 2$. But
  this rules out $q(F)<3/7.$ Therefore 
  \[q(F)=3/7.\]

  Now consider the almost free matroid \CM on set $E$ given by 
  \[\CM=\left\{ I\subseteq E\colon |I\cap \{-3,-2,-1\}|\leq 2\right\}.\]
  As it it contains all independent sets of $\CI$ except for $\{-3,-2,-1\}$, clearly
  $\CM\subseteq \CI.$ 
  Additionally we note that $\CM=U^2_{\{-3,-2,-1\}}\oplus
  U^5_{\{1,2,3,4,5\}}$, hence \CM is a matroid, in fact an inner
  matroid of \CI.

  Further, by construction $\CM=\CI-\{\{-3,-2,-1\}\}$. Therefore, all sets of
  $\CI$ but $\{-3,-2,-1\}$ are independent in $\CM$ and therefore
  \emph{perfectly} approximated. Now, clearly the best approximations of the
  set $\{-3,-2,-1\}$ from \CI are all of its two-element subsets $\{1,2\},\{2,3\},\{1,3\}$  of
  approximation quality $2/3$. Therefore $\rho(\CI,\CI,\CM)=2/3$ and
  $\rho^M(\CI,\CI)\geq
  \rho(\CI,\CI,\CM)=2/3.$

  Finally, we want to argue that $\rho^M(\CI,\CI)=2/3$. Otherwise,
  there would exist a matroid $\CMp\subseteq \CI$ with
  $\rho(\CI,\CI,\CMp)>2/3$. Since $\{-3,-2,-1\}\in\CI$  and the
  approximation-ratio is assumed to be greater than $2/3$ we can conclude that
  $\{-3,-2,-1\}\in\CMp.$ Similarly,  $\{1,2,3,4,5\}\in\CI$ implies
  that a subset of it of size greater $5*2/3$, has to be contained in
  $\CMp$. Without loss of generality suppose $\{1,2,3,4\}\in\CMp.$ As
  $\CMp$ is a matroid and contains independent sets $\{-3,-2,-1\}$ and
  $\{1,2,3,4\}$ of different size, there has to exist an element, say
  $1$, from the second set, so that
  $\{-3,-2,-1,1\}\in\CMp$. However, because of
  $\{-3,-2,-1,1\}\notin\CI$ this is impossible.
  This demonstrates 
  \[\rho^M(\CI, \CI) = 2/3.\]

  Therefore, when running the plain greedy algorithm on the
  constructed matroid $\CM$ for arbitrary objective functions, we
  obtain a performance guarantee of $2/3$ which is strictly better
  than the $1/2$-guarantee of the usual knapsack-greedy algorithm.
\end{proof}

The previous two examples, in particular the one from
Theorem~\ref{thm:smart_zero_knowledge}, raise the question of how to
\emph{easily} find a good inner matroid. We have already described an integer program~\eqref{IP} to do so, which is clearly impractical. Therefore, we
outline a heuristic that could,  in some cases,  find a good (but not best) inner matroid, by further restricting the
circuits of the independence system.

\begin{rem}
 Consider again  $E=\{-2,-1,1,2,3,4\}$ and 
  $\CI=\left\{ I\subseteq E\colon \text{ either } I=\{-2,-1\} \text{
      or } |I\cap \{-2,-1\}|\leq 1\right\}$ from the proof of Theorem~\ref{thm:smart_zero_knowledge}.
  The circuits of $\CI$ are
  $\CC=\{\{-2,-1,i\}\colon i\in \{1,2,3,4\}\}.$
  If \CI were a matroid, then, strong circuit elimination  should
  hold.
  I.e., strong circuit elimination for $\{-2,-1,1\}$ and $\{-2,-1,2\}$ would
  require that there exits circuits in $\{-2,-1,1,2\}$
  \begin{inparaitem}
  \item containing $1$ and avoiding $-2$,
  \item containing $2$ and avoiding $-2$,
  \item containing $1$ and avoiding $-1$, and
  \item containing $2$ and avoiding $-1$.
  \end{inparaitem}
  In comparison with $\CC$ we see, that no circuit in
  $\{-2,-1,1,2\}$ contains $1$ and avoids $-2$ and so on.
To get ``closer'' to an inner matroid, we could
  either drop one of the problematic circuits, say $\{-2,-1,1\}$ but
  this would enlarge the independence system, the `opposite' to an
  \emph{inner matroid}, or we could
  \begin{compactdesc}
  \item[add missing circuits:] for example above we could add circuits
    $\{1,-1\},\{2,-1\},\{1,-2\},\{2,-2\}$; this way we would get
    closer to an inner matroid, while some of the original circuits we started
    with, like $\{-2,-1,1\}$ would remain dependent, but would no
    longer be minimal dependent sets, as they contain some of the newly minted
    circuits, like $\{1,-2\}$. If we would add all circuits that could
    be derived in this fashion, we would, indeed end with an inner
    matroid, that turns out to be $U^1_E.$ However, it provides only
    a Milgrom approximation guarantee of $1/5$ worse than the
    rank-quotient of $1/2$. Or we could
  \item[``destroy'' the offending circuits, by adding circuits
    contained in them:] an obvious way to make the offending circuits
    $\{-2,-1,1\}$ and $\{-2,-1,2\}$ go away is to add a circuit
    $\{-2,-1\}$. In fact all previous circuits contain this one, so
    that the resulting circuit system would be $\CC'=\{\{-2,-1\}\}$. Indeed, the
    independence system induced by this single circuit is just the
    inner matroid $\CM=U^1_{\{-2,-1\}}\oplus
    U^4_{\{1,2,3,4\}}$ from the proof of
    Theorem~\ref{thm:smart_zero_knowledge} realizing the approximation
    guarantee of $1/2$.
  \end{compactdesc}
\end{rem}

\begin{rem}
  One might wonder, what else can be said about the independence
  system $\CI$ of Theorem~\ref{thm:smart_zero_knowledge}. As it turn
  out $\CI\contract \{-2\}$ is  one of the previously defined
  twin-peaks independence systems with ground sets of size $1$ and $4$.
\end{rem}

The zero knowledge assumption is weaker than necessary. If $v\geq
\bz$, we know that a maximum weight independent set must be, unsurprisingly, a basis of $E$. %
Thus, a more reasonable knowledge assumption is $\CO=\CB_\CI$. %
This provides another possibility to improve over simple greedy for
independence systems.

Suppose $v\geq \bz$ and $\CO \subset \CB_\CI$. In this case, prior
information rules out some basis from being optimal. To ensure an
apples with apples comparison between the greedy algorithm and the
algorithm in \citet{milgrom-2017} we must allow the greedy algorithm
to make use of the information contained in $\CO$ as well. This may
require modifying the greedy algorithm so much that it become an
entirely different algorithm altogether.  A straightforward way to do
this is to apply the greedy algorithm to an independence system
$(E,\COd)$ of sets contained in elements of $\CO$. The example we
describe below shows that the greedy algorithm will be dominated by
Milgrom's heuristic. Under one interpretation of $\CO$ this
is not surprising. Suppose, that $\CI \setminus \CO$ is the collection
of sets that can never be optimal. Then, $\COd$ would include sets
that should definitely not be selected.

\begin{example}\label{ex:stab}
Consider the stable set problem (largest set of pairwise nonadjacent
vertices) on a path of $9$ vertices:
\vspace*{-5mm}\begin{center}
  \begin{tikzpicture}[scale=1.2,style=thin]%
    \path[clip](-1,0.5) rectangle (9,1.5);

    \tikzstyle{every node}=[font={\small}]; \foreach \x/\y/\i in {
      0.5/1/1, 1.5/1/2, 2.5/1/3, 3.5/1/4, 4.5/1/5, 5.5/1/6, 6.5/1/7,
      7.5/1/8, 8.5/1/9%
    }{%
      \node[circle, draw, fill=white,%
      inner sep=2pt, minimum width=14pt] at (\x,\y)(\i){\i}; }
    \begin{scope}[black]%
      \foreach \i/\j in { 1/2, 2/3, 3/4, 4/5, 5/6, 6/7, 7/8, 8/9%
      }{ \draw (\i)--(\j); }
    \end{scope}%
  \end{tikzpicture}
\end{center}\vspace*{-5mm}
The underlying independence system has
$E=\{1,2,3,4,5,6,7,8,9\}$ and the stable sets form the set-system
$\CI$.
The maximum cardinality basis of $\CI$ is the set $\{1,3,5,7,9\}$ of
cardinality $5$. An example of a small  basis of $E$ would be
$\{2,5,8\}$ of size $3$. For the Korte-Hausmann bound, consider the
set $\{1,2,3\}$ that demonstrates $q\leq \frac12$. 

Suppose $\CO$ is the collection of stable sets $S$ of size at most $4$
with $|S\cap \{3,6,9\}|\leq 1$. Notice, $\{1,3,5,7\} \in \CO$.

Consider the set $\{1,2,3\}$.  Assuming only subsets of elements of $\CO$, are independent, this set has a rank quotient of $1/2$, i.e., $q(\COd) = 1/2$. Hence, the greedy algorithm applied to the independence system consisting of the elements of $\COd$ will have a worst case bound of at most $1/2$.

The following matroid $\CM=U^1_{\{1,2\}}\oplus U^1_{\{4,5\}}\oplus
U^1_{\{7,8\}}$
is contained in $\CI$ since every independent set in $\CM$ is a stable
set from $\CI$. It is straightforward to verify that 
$\rho^M(\CI,\CO) = 3/4 > 1/2$. Interestingly, the largest set of $\CI$
is not independent in the matroid $\CM$. 
\end{example}

This example highlights the essential difference between Milgrom's heuristic and the greedy algorithm. The first is looking for a basis that has large overlap with each set in $\CO$. The greedy algorithm seeks a basis that has a large overlap with \emph{every} set that is contained in $\COd$.

\section{Discussion}
Our reading of \cite{milgrom-2017} suggests that Milgrom proposed the existence of a `good' inner matroid as an explanation for the success of the greedy algorithm in some settings. He writes:
\begin{quote}
``In practice, procedures based on greedy algorithms often perform very well.''
\end{quote}
Our analysis shows that existence of a good inner matroid by
itself is not sufficient to explain the practical observation.
It is possible for there to be no good inner matroid approximation,
yet the greedy algorithm performs well. Further, there are examples
where the inner matroid approximation will dominate that of the greedy
algorithm.

Certainly, the prior knowledge encoded in $\CO$ is beneficial.  
The idea of incorporating prior information about the optimal solution into an optimization problem is not new. There are three approaches we are aware of: probabilistic, uncertainty sets and stability. The first encodes the prior information in terms of a probability distribution over the possible objective value coefficients. (This is sometimes relaxed to a class of distributions sharing common moments.) The second, assumes that the objective coefficients are drawn from some set given a-priori. Stability assumes that the optimal solution of the instance under examination does not change under perturbations to the objective function (see \cite{bilulinial}).
The novelty of Milgrom's proposal is to encode the prior information  in a description of the set of possible optima.

However, as Theorem~\ref{thm:smart_zero_knowledge} makes clear, even in the absence of prior knowledge the inner matroid approximation can outperform greedy. This is because the inner matroid approximation exploits \emph{global} information about the structure of the independence system. Roughly speaking, the inner matroid approximation `kills' off low rank bases so as to prevent greedy from being stuck there. If the optimal basis is not to be found among these low rank sets, no harm is done. If not, precisely because they are low rank, one can approximate the optimal objective function value well using only a subset of the elements of these low rank sets.

\subsection*{Acknowledgements}
The authors want to thank a referee for helpful comments and finding a
gap in a previous proof, that moved us to state and prove
Theorem~\ref{thm:smart_zero_knowledge}.

\bibliographystyle{sdv-plainnat_nop}

\newcommand{\urlprefix}{}%
{%
\small \bibliography{greedy}}

\begin{thebibliography}{6}
\expandafter\ifx\csname natexlab\endcsname\relax\def\natexlab#1{#1}\fi
\expandafter\ifx\csname url\endcsname\relax
  \def\url#1{\texttt{#1}}\fi
\expandafter\ifx\csname urlprefix\endcsname\relax\def\urlprefix{URL }\fi
\providecommand{\selectlanguage}[1]{\relax}

\bibitem[{Bilu \protect\BIBand{} Linial(2010)}]{bilulinial}
Bilu, Yonatan \protect\BIBand{} Nathan Linial (2010).
\newblock Are stable instances easy?
\newblock In \emph{Innovations in Computer Science - {ICS} 2010, Tsinghua
  University, Beijing, China, January 5-7, 2010. Proceedings} (edited by
  Andrew~Chi{-}Chih Yao). Tsinghua University Press,  332--341.

\bibitem[{Hausmann et~al.(1980)Hausmann, Korte, \protect\BIBand{}
  Jenkyns}]{hausmann-korte-jenkyns-1980}
Hausmann, Dirk, Bernhard Korte, \protect\BIBand{} Tom~A. Jenkyns (1980).
\newblock Worst case analysis of greedy type algorithms for independence
  systems.
\newblock In \emph{Combinatorial Optimization I} (edited by Manfred~W.
  Padberg), volume~12 of \emph{Mathematical Programming Studies}. Springer
  Verlag, Berlin Heidelberg,  120--131.

\bibitem[{Jenkyns(1976)}]{jenkyns-1976}
Jenkyns, T.~A. (1976).
\newblock The efficiency of the greedy algorithm.
\newblock In \emph{Proceedings of the 7th Southeast Conference on
  Combinatorics, Graph Theory, and Computing}, Congressus Numerantium. Utilitas
  Mathematica, Winnipeg,  341--350.

\bibitem[{Korte \protect\BIBand{} Hausmann(1978)}]{korte-hausmann-1978}
Korte, Bernhard \protect\BIBand{} Dirk Hausmann (1978).
\newblock An analysis of the greedy algorithm for independence systems.
\newblock In \emph{Algorithmic Aspects of Combinatorics} (edited by B.~Alspach,
  P.~Hell, \protect\BIBand{} D.J. Miller), volume~2 of \emph{Annals of Discrete
  Mathematics}. North-Holland, Amsterdam,  65--74.

\bibitem[{Milgrom(2017)}]{milgrom-2017}
Milgrom, Paul (2017).
\newblock \emph{{D}iscovering {P}rices: {A}uction {D}esign in {M}arkets with
  {C}omplex {C}onstraints}.
\newblock Kenneth J. Arrow lecture Series. Columbia University Press, New York
  and Chichester, West Sussex.

\bibitem[{Oxley(1992)}]{oxley-b1992}
Oxley, James~G. (1992).
\newblock \emph{{M}atroid {T}heory}.
\newblock Oxford University Press, Oxford, GB.

\end{thebibliography}

\end{document}